\documentclass[a4paper,USenglish,autoref,cleveref,thm-restate]{lipics-v2021}

\usepackage[utf8]{inputenc}
\usepackage[T1]{fontenc}

\usepackage{amssymb,amsmath}
\usepackage{graphicx}
\usepackage{enumerate}
\usepackage{tikz}
\usetikzlibrary{patterns,snakes,shapes}

\usepackage{xcolor}
\usepackage{mathtools}
\usepackage{microtype}
\usepackage{amsfonts}
\usepackage{comment}
\usepackage{mathrsfs}
\usepackage{blindtext}
\usepackage{nicefrac}

\usepackage{trimspaces}
\usepackage{nccfoots}
\usepackage[absolute]{textpos}
\usepackage{cite}
\usepackage{enumerate}
\usepackage[backgroundcolor = blue!50]{todonotes}
\usepackage{longtable}

\definecolor{blue}{rgb}{0.1,0.2,0.5}
\definecolor{brown}{rgb}{0.6,0.6,0.2}
\usepackage{enumerate}
\usepackage{latexsym}
\usepackage{multicol}

\theoremstyle{plain}
\crefformat{cthmin}{#2Theorem~#1#3}
\Crefformat{cthmin}{#2Theorem~#1#3}

\crefformat{ccorin}{#2Corollary~#1#3}
\Crefformat{ccorin}{#2Corollary~#1#3}
\newcommand{\p}{\mathsf{p}}

\newcommand{\NP}{$\mathsf{NP}$}

\DeclareMathOperator{\dist}{dist}
\DeclareMathOperator{\diam}{diam}

\newcommand{\Prob}{\mathsf{Pr}}
\newcommand{\Exp}{\mathsf{E}}

\newcommand{\N}{\mathbb{N}}

\renewcommand{\phi}{\varphi}
\renewcommand{\epsilon}{\varepsilon}

\newcommand{\Oh}{\mathcal{O}}

\renewcommand{\leq}{\leqslant}
\renewcommand{\geq}{\geqslant}

\newcommand{\col}[1]{{#1}-\textsc{Coloring}\xspace}
\newcommand{\lcol}[1]{\textsc{List} {#1}-\textsc{Coloring}\xspace}

\newcommand{\wei}{\mathfrak{w}}

\title{Faster 3-coloring of small-diameter graphs}

\author{Micha\l{} D\k{e}bski}{Warsaw University of Technology, Faculty of Mathematics and Information Science}{m.piecyk@mini.pw.edu.pl}{}{}
\author{Marta Piecyk}{Warsaw University of Technology, Faculty of Mathematics and Information Science}{m.piecyk@mini.pw.edu.pl}{}{}
\author{Pawe\l{} Rz\k{a}\.zewski}{Warsaw University of Technology, Faculty of Mathematics and Information Science\\
\& University of Warsaw, Institute of Informatics}{p.rzazewski@mini.pw.edu.pl}{https://orcid.org/0000-0001-7696-3848}{}

\authorrunning{M. D\k{e}bski, M. Piecyk, and P. Rz\k{a}\.zewski}

\Copyright{M. D\k{e}bski, M. Piecyk, and P. Rz\k{a}\.zewski}

\funding{Supported by Polish National Science Centre grant no. 2018/31/D/ST6/00062.}
\acknowledgements{The authors are sincerely grateful to Carla Groenland for many inspiring discussions and useful comments on the manuscript.}

\nolinenumbers
\hideLIPIcs
\begin{document}

\maketitle
\begin{abstract}
We study the 3-\textsc{Coloring} problem in graphs with small diameter.
In 2013, Mertzios and Spirakis showed that for $n$-vertex diameter-2 graphs this problem can be solved in subexponential time $2^{\mathcal{O}(\sqrt{n \log n})}$.
Whether the problem can be solved in polynomial time remains a well-known open question in the area
of algorithmic graphs theory. 

In this paper we present an algorithm that solves 3-\textsc{Coloring} in $n$-vertex diameter-2 graphs in time $2^{\mathcal{O}(n^{1/3}  \log^{2} n)}$. This is the first improvement upon the algorithm of Mertzios and Spirakis in the general case, i.e., without putting any further restrictions on the instance graph.

In addition to standard branchings and reducing the problem to an instance of 2-\textsc{Sat}, the crucial building block 
of our algorithm is a combinatorial observation about 3-colorable diameter-2 graphs, which is proven using a probabilistic argument.

As a side result, we show that 3-\textsc{Coloring} can be solved in time $2^{\mathcal{O}( (n \log n)^{2/3})}$ in $n$-vertex diameter-3 graphs.
We also generalize our algorithms to the problem of finding a list homomorphism from a small-diameter graph to a cycle.
\end{abstract}
\newpage

\section{Introduction}
For many \NP-hard graph problems, the instances constructed in hardness reductions are very specific and ``unstructured''.
Thus a natural direction of research is to study how additional restrictions imposed on the input graphs affect the complexity of the problem.
In particular, we would like to understand if the additional knowledge about the structure of the instance makes the problem easier,
and what are the ``minimal'' sets of restrictions that we need to impose in order to make the problem efficiently solvable.

Usually, the main focus in the area is on \emph{hereditary} classes of graphs, i.e., classes that are closed under vertex deletion.
Prominent examples are perfect graphs~\cite{spgt,GROTSCHEL1984325}, graphs excluding a certain induced subgraph~\cite{DBLP:journals/jgt/GolovachJPS17} or minor~\cite{DBLP:conf/focs/DemaineHK05}, and intersection graphs of geometric objects~\cite{DBLP:conf/compgeom/Kratochvil11}.
Studying these classes has led to a better understanding of the structure of such graphs~\cite{DBLP:journals/corr/abs-1901-00335,DBLP:journals/combinatorica/ChudnovskyS10,DBLP:journals/jct/KratochvilM94,DBLP:journals/jct/RobertsonS86} and a discovery of numerous exciting algorithmic techniques~\cite{DBLP:journals/jacm/BonamyBBCGKRST21,DBLP:journals/siamcomp/BergBKMZ20,DBLP:reference/algo/FominDHT16,DBLP:conf/focs/GartlandL20,DBLP:conf/esa/MarxP15}.
Let us point out that the property of being hereditary is particularly useful in the construction of recursive algorithms based on branching or the divide \& conquer paradigm.

However, there are many natural classes of graphs that are not hereditary, for example graphs with bounded diameter.
Such graphs are interesting not only for purely theoretical reasons: for example social networks tend to have small diameter~\cite{DBLP:journals/socnet/Schnettler09}.

Observe that for any graph $G$, a graph $G^+$ obtained from $G$ by adding a universal vertex has diameter 2.
Since the graph $G$ may be arbitrarily complicated, the fact that $G^+$ has small diameter does not imply that its structure is simple.
This observation can be used to show that many classic computational problems are \NP-hard for graphs of bounded diameter and they cannot be solved in subexponential time under the ETH. For instance, the size of a maximum independent set in $G^+$ is equal to the size of a maximum independent set in $G$, and thus \textsc{Max Independent Set} in diameter-2 graph is \NP-hard and cannot be solved in subexponential time, unless the ETH fails.

A similar argument applies to \col{$k$}: the graph $G^+$ is $k$-colorable if and only if $G$ is $(k-1)$-colorable.
Thus, for any $k \geq 4$, the \col{k} problem is \NP-hard and admits no subexponential-time algorithm (under the ETH) in diameter-2 graphs.
However, the reasoning above breaks down for $k=3$, as \col{2} is polynomial-time solvable.

This peculiar open case was first studied by Mertzios, Spirakis~\cite{DBLP:journals/algorithmica/MertziosS16} who proved that the problem can be solved in \emph{subexponential time}. The result holds even for the more general 
\lcol{3} problem, where each vertex $v$ of the instance graph is equipped with a \emph{list} $L(v) \subseteq \{1,2,3\}$, and we ask for a proper coloring, in which every vertex gets a color from its list.

\begin{theorem}[Mertzios, Spirakis~\cite{DBLP:journals/algorithmica/MertziosS16}]\label{thm:MSdiam2}
The \lcol{3} problem on $n$-vertex graphs with diameter 2 can be solved in time $2^{\Oh(\sqrt{n \cdot \log n})}$.
\end{theorem}

Their algorithm is based on a simple win-win argument.
The first ingredient is a well-known fact that every graph with $n$ vertices and minimum degree $\delta$ has a dominating set of size $\Oh \left (\frac{n \log \delta}{\delta} \right)$~\cite[Theorem~1.2.2]{DBLP:books/daglib/0021015}.
On the other hand, in a diameter-2 graph, the neighborhood of each vertex is a dominating set, so there is a dominating set of size $\delta$.
Thus, every diameter-2 graph has a dominating set $S$ of size $\Oh \left( \min(\delta, \frac{n \log \delta}{\delta}) \right)$ which is upper-bounded by $\Oh(\sqrt{n \log n})$.

We exhaustively guess the coloring of vertices in $S$ and update the lists of their neighbors.
Note that after this, each uncolored vertex has at least one colored neighbor, and thus each list has at most 2 elements.
A classic result by Edwards~\cite{DBLP:journals/tcs/Edwards86} shows that such a problem can be solved in polynomial time by a reduction to 2-\textsc{Sat}. Summing up, the complexity of the algorithm is bounded by $2^{|S|} \cdot n^{\Oh(1)} = 2^{\Oh( \sqrt{n \log n})}$.

Let us point out that the bound $\sqrt{n}$ appears naturally for different parameters of diameter-2 graphs,
for example the maximum degree of such a graph is $\Omega(\sqrt{n})$. 
Based on this, one can also construct different algorithms for \lcol{3} in diameter-2 graphs with running time matching the one of \cref{thm:MSdiam2} (see \cref{sec:diam3}).

If it comes to \col{3} in diameter-3 graphs, Mertzios and Spirakis~\cite{DBLP:journals/algorithmica/MertziosS16} proved that the problem is \NP-hard, but their reduction is quadratic. Thus, under the ETH, the problem cannot be solved in time $2^{o(\sqrt{n})}$.
Actually, the authors carefully analyzed how the lower bound depends on the minimum degree of the input graph,
and presented three hardness reductions, each for a different range of $\delta$.
Furthermore, they showed that the problem can be solved in time $2^{\Oh \left( \min (\delta \cdot \Delta, \frac{n \log \delta}{\delta}) \right)}$, where $\Delta$ is the maximum degree.
The argument again follows from the observation that each diameter-3 graph has a dominating set of size at most $\delta \cdot \Delta$.
Let us point out that if $\Delta = \Theta(n)$ and $\delta = O(1)$, then the running time is exponential in $n$.
In~\cref{fig:diam3} we summarize the results for diameter-3 graphs with given minimum degree.

\begin{figure}
\centering
\includegraphics[scale=0.8]{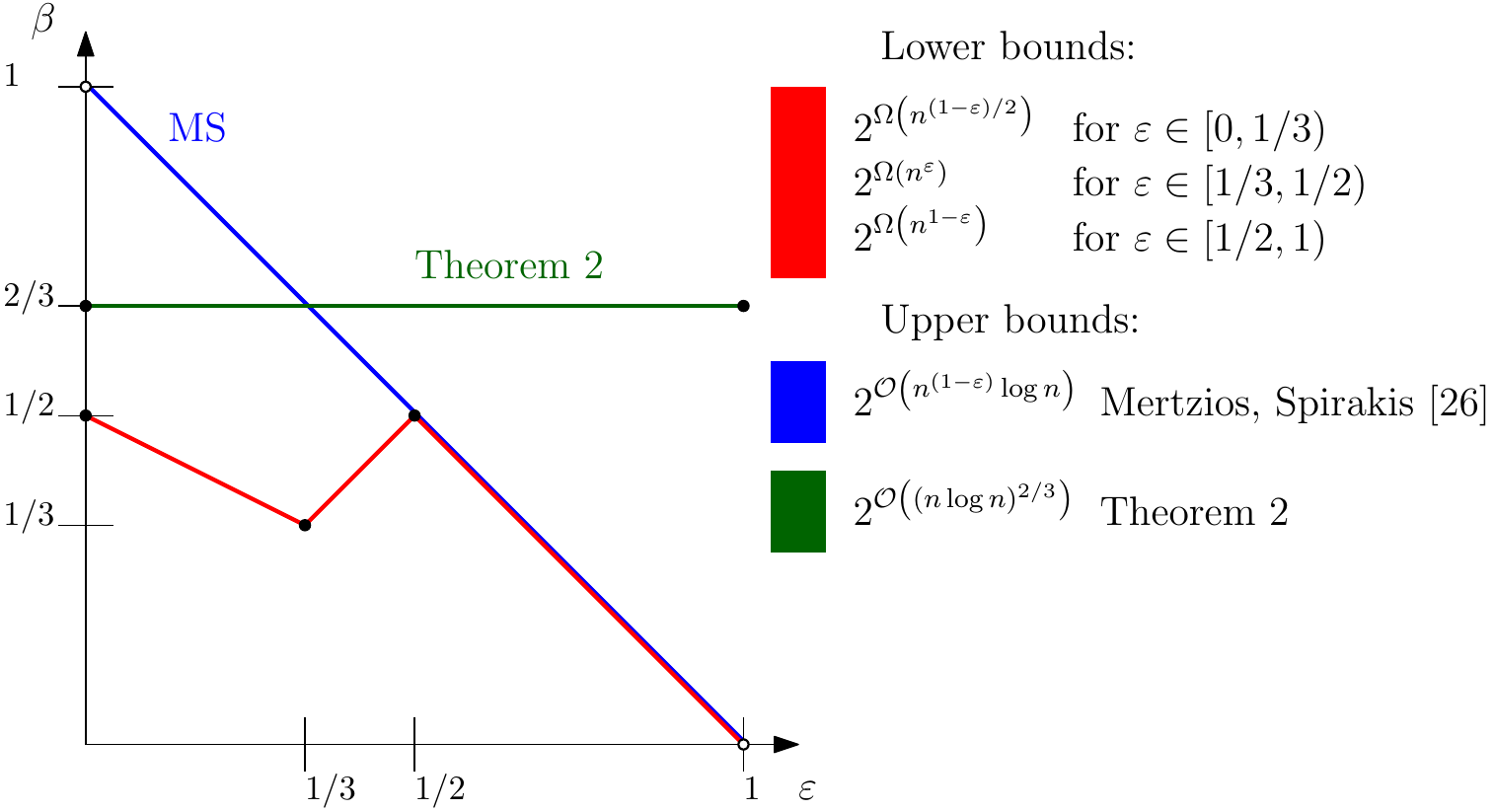}
\caption{The complexity of \lcol{3} in $n$-vertex diameter-3 graphs with minimum degree $\Theta(n^\varepsilon)$ for $\varepsilon \in [0,1]$.
The complexity bound is of the form $2^{\Oh(n^{\beta} \cdot \log^{\Oh(1)} n)}$ for $\beta \in [0,1]$.
}\label{fig:diam3}
\end{figure}

The story stops at diameter 3: a textbook reduction from \textsc{NAE-Sat} to \col{3} builds a graph with diameter 4 and number of vertices linear in the size of the formula~\cite[Theorem~9.8]{DBLP:books/daglib/0072413}.
This proves that the \col{3} problem in diameter-4 graphs is \NP-hard and cannot be solved in subexponential time, unless the ETH fails.

Closing the gaps left by Mertzios and Spirakis~\cite{DBLP:journals/algorithmica/MertziosS16}, and in particular determining the complexity of \col{3} in diameter-2 graphs, is a notorious open problem in the area of graph algorithms.
We know polynomial-time algorithms if some additional restrictions are imposed on the instance~\cite{DBLP:conf/mfcs/MartinPS19,DBLP:journals/corr/abs-2101-07856}.
However, to the best of our knowledge, no progress in the general case has been achieved.

Let us also point out that some other problems, including different variants of graph coloring, have also been studied for small-diameter graphs~\cite{DBLP:journals/ipl/BonamyDFJP18,DBLP:journals/corr/abs-2004-11173,graphmod,DBLP:journals/corr/abs-2104-10593}.

\subparagraph*{Our results.}
As our first result, in \cref{sec:diam3} we show a simple subexponential-time algorithm for the \lcol{3} problem in diameter-3 graphs.
\begin{theorem}\label{thm:diam3}
The \lcol{3} problem on $n$-vertex graphs with diameter 3 can be solved in time $2^{\Oh(n^{2/3} \cdot \log^{2/3} n)}$.
\end{theorem}
Note that the running time bounds does not depend on the maximum nor the minimum degree of the input graph.
In particular, this is the \emph{first} algorithm for \lcol{3}, whose complexity is subexponential for \emph{all} diameter-3 graphs, see~\cref{fig:diam3}.

Let us present a high-level overview of the proof.
We partition the vertex set of our graph into three sets $V_1,V_2,V_3$, where $V_i$ contains the vertices with lists of size $i$.
If the graph contains a vertex $v \in V_2 \cup V_3$ with at least $n^{1/3}$ neighbors in $V_2 \cup V_3$,
then we can effectively branch on the color of $v$.
Otherwise, we observe that for any $v \in V_2 \cup V_3$, the set $S$ of vertices at distance at most 2 from $v$ in the graph induced by sets $V_2 \cup V_3$ dominates $V_3$, i.e., every vertex from $V_3$ is in $S$ or has a neighbor in $S$.
Thus, after exhaustively guessing the coloring of $S$, all lists are reduced to size at most 2 and then we can finish in polynomial time, using the already-mentioned result of Edwards~\cite{DBLP:journals/tcs/Edwards86}.

\smallskip
In \cref{sec:diam2} we prove the following theorem, which is the  main result of the paper.
\begin{restatable}{theorem}{mainthm}
\label{thm:main}
The \lcol{3} problem on $n$-vertex graphs with diameter 2 can be solved in time $2^{\Oh(n^{1/3} \cdot \log^{2} n)}$.
\end{restatable}

Again, let us give some intuition about the proof.
We partition the vertex set of $G$ into $(V_1,V_2,V_3)$, as previously. We aim to empty the set $V_3$, as then the problem can be solved in polynomial time.
We start with applying three branching rules.
The first one is similar as in the proof of \cref{thm:diam3}: if we find a vertex $v$ with many neighbors in $V_3$, we can branch on choosing the color of $v$. The other two branching rules are somewhat technical and their purpose is not immediately clear, so let us not discuss them here.

The main combinatorial insight that is used in our algorithm is as follows.
Consider an instance $(G,L)$, where $G$ is of diameter 2 and none of the previous branching rules can be applied.
Suppose that $G$ has a proper 3-coloring $\phi$ that respects lists $L$.
Then there is a color $a \in \{1,2,3\}$ and sets $S \subseteq V_3 \cap \varphi^{-1}(a)$ and  $\widetilde{S} \subseteq V_3 \setminus \varphi^{-1}(a)$,
each of size $\Oh(n^{1/3} \log n)$, with the following property: 
\begin{description}
\item[($\star$)] $S \cup \widetilde{S} \cup \left( N(S) \cap N(\widetilde{S}) \right)$ dominates at least $\frac{1}{6}$-fraction of $V_3$,
\end{description}
where $N(S)$ (resp. $N(\widetilde{S})$) denotes the set of vertices with a neighbor in $S$ (resp. $\widetilde{S}$).
The existence of the sets $S$ and $\widetilde{S}$ is shown using a probabilistic argument.

Now we proceed as follows. We enumerate all pairs of disjoint sets $S$ and $\widetilde{S}$, each of size $\Oh(n^{1/3} \log n)$.
If they satisfy the property ($\star$), we exhaustively guess the color $a$ used for every vertex of $S$ and the coloring of $\widetilde{S}$ with colors $\{1,2,3\} \setminus \{a\}$. Then we update the lists of the neighbors of colored vertices.
Note that the color of every vertex from $N(S) \cap N(\widetilde{S})$ is now uniquely determined.
Thus, for at least $\frac{1}{6}$-fraction of vertices $v \in V_3$, they are either already colored or have a colored neighbor and thus their lists are of size at most 2.
Thus our instance was significantly simplified and we can proceed recursively.

\smallskip
Finally, in \cref{sec:ext} we investigate possible extensions of our algorithms to some generalizations of (\textsc{List}) \col{3}.
We observe that our approach can be used to obtain subexponential-time algorithms for the problem of finding a \emph{list homomorphism} from a graph with diameter at most 3 to certain graphs, including in particular all cycles.
We refer to \cref{sec:lhom} for the definition of the problem and the precise statement of our results;
let us just point out that under the ETH the problems considered there cannot be solved in subexponential time in general graphs~\cite{DBLP:journals/combinatorica/FederHH99,DBLP:journals/jgt/FederHH03}

We conclude with discussing the possibility of extending our algorithms to weighted coloring problems, with \textsc{Independent Odd Cycle Transversal}~\cite{BonamyDFJP19} as a prominent special case.

\section{Preliminaries}
For an integer $n$, we denote $[n] := \{1,2,\ldots,n\}$. For a set $X$, by $2^{X}$ we denote the family of all subsets of $X$.
All logarithms in the paper are natural.

Let $G = (V,E)$ be a connected graph.
For two vertices $u$ and $v$, by $\dist_G(u,v)$ we denote the distance from $u$ to $v$, i.e., the number of edges on a shortest $u$-$v$ path in $G$.
The \emph{diameter} of $G$, denoted by $\diam(G)$, is the maximum value of $\dist(u,v)$ over all $u,v \in V$.

For a vertex $v$, by $N_G(v)$ we denote its \emph{open neighborhood}, i.e., the set of all vertices adjacent to $v$.
The \emph{closed neighborhood} of $v$ is defined as $N_G[v] := N_G(v) \cup \{v\}$.
For an integer $p$, by $N^{\leq p}_G[v]$ we denote the set of vertices at distance at most $p$ from $v$, and define $N^{\leq p}_G(v) := N^{\leq p}_G[v] \setminus \{v\}$.
For a set $X$ of vertices, we define $N_G(X) := \bigcup_{v \in X} N_G(v) \setminus X$ and $N_G[X] := N_G(X) \cup X$.
For sets $X,Y \subseteq V$, we say that $X$ \emph{dominates} $Y$ if $Y \subseteq N_G[X]$.
By $\deg_G(v)$ we denote the \emph{degree} of a vertex $v$, i.e., $|N_G(v)|$.

If the graph $G$ is clear from the context, we drop the subscript in the notation above and simply write $\dist(u,v)$, $N(v)$, etc.
By $\Delta(G)$ we denote the maximum vertex degree in~$G$.

The following result by Edwards~\cite{DBLP:journals/tcs/Edwards86} will be an important tool used in all our algorithms.

\begin{theorem}[Edwards~\cite{DBLP:journals/tcs/Edwards86}]\label{thm:edwards}
Let $G=(V,E)$ be a graph and let $L : V \to 2^\N$ be a list assignment, such that for every $v \in V$ it holds that $|L(v)| \leq 2$.
Then in polynomial time we can decide whether $G$ admits a proper vertex coloring that respects lists $L$.
\end{theorem}

\subparagraph{Reduction rules.} Let $(G,L)$ be an instance of the \lcol{3} problem. It is straightforward to observe that the following reduction rules can be safely applied, as they do not change the set of solutions. Moreover, each of them can  be applied in polynomial time.
\begin{enumerate}[{R}1]
    \item If there exists a vertex $v$ such that $L(v)$ contains only one color $a$, then remove $a$ from $L(u)$ for each vertex $u\in N(v)$.
    \item If there exists a vertex $v$ such that $|L(v)|=0$, then report failure.
    \item If $|L(v)|\leq 2$ for each vertex $v$, then solve the problem using \cref{thm:edwards}.
\end{enumerate}
An instance $(G,L)$ for which none of the reduction rules can be applied is called \emph{reduced}.
Note that the reduction rules do not remove any vertices from the graph, even if their color is fixed. This is because such an operation might increase the diameter.
\subparagraph{Layer structure.} 
Let $(G,L)$ be a reduced instance of \lcol{3}. For $i \in [3]$, let $V_i$ be the set of vertices $v$ of $G$, such that $|L(v)|=i$.
Note that $(V_1,V_2,V_3)$ is a partition of $V$; we will call it the \emph{layer structure} of $G$.
Observe that since R1 cannot be applied to $(G,L)$, it holds that $N(V_1) \subseteq V_2$, i.e., there are no edges between $V_1$ and $V_3$.

We conclude this section with an important observation about layer structures of graphs with diameter at most 3.
\begin{proposition} \label{prop:distance}
Let $(G,L)$ be a reduced instance of \lcol{3}, where $G$ has diameter $d \leq 3$,
and let $(V_1,V_2,V_3)$ be the layer structure of $G$. Then, for any $u,v \in V_2\cup V_3$, at least one of the following hold:
\begin{enumerate}[a)]
\item $u$ and $v$ are at distance at most $d$ in $G[V_2 \cup V_3]$, or
\item $\{u,v\} \cap V_2 \neq \emptyset$.
\end{enumerate}
\end{proposition}
\begin{proof}
If $V_1 = \emptyset$, then the first outcome follows, since $G = G[V_2 \cup V_3]$.
So assume that $V_1 \neq \emptyset$. Consider $u,v \in V_{3}$ and suppose that they are not at distance at most $d$ in $G[V_2 \cup V_3]$.
Since they are at distance at most $d$ in $G$, all shortest $u$-$v$-paths in $G$ must intersect $V_1$.
However, for any $x \in V_1$, it holds that $\dist(u,x) \geq 2$ and $\dist(v,x) \geq 2$.
Thus $\dist(u,v) \geq 4$, contradicting the fact that $\diam(G) \leq 3$.
\end{proof}

Observe that \cref{prop:distance} does not generalize to diameter-4 graphs: consider e.g. 5-vertex path $P_5$ with consecutive vertices $v_1,v_2,v_3,v_4,v_5$, where $V_1 = \{v_3\}$.
Vertices $v_1$ and $v_5$ are in $V_3$, they are at distance 4 in $P_5$, but not in $P_5[V_2 \cup V_3] = P_5 - \{v_3\}$.

\cref{prop:distance} immediately yields the following corollary.
\begin{corollary}\label{cor:neighborhoods}
Let $(G,L)$ be an instance of the \lcol{3}, where $G$ has diameter $d \in \{2,3\}$,
and let $(V_1,V_2,V_3)$ be the layer structure of $G$.
For every $v \in V_{3}$, the set $N^{\leq d-1}_{G[V_2 \cup V_3]}[v]$ dominates $V_{3}$.
\end{corollary}

\section{Coloring diameter-3 graphs}\label{sec:diam3}
In this section we present a simple proof of \cref{thm:diam3}.
Actually, we will show the following more general result, which yields yet another $2^{{\Oh}(\sqrt{n \log n})}$-algorithm for diameter-2 graphs.
This will serve as a warm-up before showing our main result, i.e., \cref{thm:main}.

\begin{theorem}
The \lcol{3} problem on $n$-vertex graphs $G$ can be solved in time:
\begin{enumerate}
\item $2^{\Oh( n^{1/2} \log^{1/2} n)}$, if $\diam(G) = 2$,
\item $2^{\Oh( n^{2/3} \log^{2/3} n )}$, if $\diam(G) = 3$.
\end{enumerate}
\end{theorem}
\begin{proof}
Let $(G,L)$ be an instance of \lcol{3}, where $G$ has $n$ vertices and diameter $d \in \{2,3\}$.
Without loss of generality we may assume that it is reduced.
Let $(V_1,V_2,V_3)$ be the layer structure of $(G,L)$ and let us define a measure $\mu := 2|V_2| + 3|V_3|$.

First, consider the case that there is a vertex $v \in V_2 \cup V_3$ with at least $(\mu \log \mu)^{1/d}$ neighbors in $V_2 \cup V_2$.
Since each vertex of $V_2 \cup V_3$ has one of four possible lists, there is a subset of at least $\frac{(\mu \log \mu)^{1/d}}{4}$ neighbors of $v$ that all have the same list $L'$.
Note that there is $a \in L(v) \cap L'$ since both are subsets of size at least $2$ of a set of size $3$.
We branch on coloring the vertex $v$ with color $a$ or not. In other words, in the first branch we remove from $L(v)$ all elements but $a$, and in the other one we remove $a$ from $L(v)$. Note that after reducing the obtained instance, at least  $\frac{(\mu \log \mu)^{1/d}}{4}$ vertices will lose at least one element from their list in one of the two branches.

We can bound the number of instances produced by applying this step exhaustively as follows:
\[
F(\mu) \leq F \left( \mu - \frac{(\mu \log \mu)^{1/d}}{4} \right) + F(\mu-1).
\]
Solving this inequality, we obtain that $F(\mu) = \mu^{\Oh \left (\frac{\mu}{(\mu \log \mu)^{1/d}} \right)}  = 2^{\Oh \left( (\mu \log \mu)^{1-1/d} \right)}$.

We can hence arrive at the case that $\Delta(G[V_2 \cup V_3])<(\mu \log \mu)^{1/d}$.
Recall that since the reduction rule R3 cannot be applied, it holds that $V_3 \neq \emptyset$. Pick any vertex $v \in V_3$.
Define $X := N^{\leq d-1}_{G[V_2 \cup V_3]}[v]$; by \cref{cor:neighborhoods}, the set $X$ dominates $V_{3}$.
Furthermore
\[|X| \leq 1 + \Delta(G[V_2 \cup V_3])^{d-1} = \Oh( (\mu \log \mu)^{(d-1)/d} ).\]
We exhaustively guess the coloring of $X$, which results in at most $3^{|X|} = 2^{\Oh \left( (\mu \log \mu)^{1-1/d} \right)}$ branches.
As $X$ dominates $V_3$, after applying the reduction rule R1 to every vertex of $X$, in each branch there are no vertices with three-element lists. Therefore, the instance obtained in each of the branches is solved in polynomial time using reduction rule R3.
The claimed bound follows since $\mu \leq 3n$.
\end{proof}

\section{Coloring diameter-2 graphs}\label{sec:diam2}
In this section we prove the main result of the paper, i.e., \cref{thm:main}.
Let us recall the following variant of the Chernoff concentration bound.

\begin{theorem}[{\cite[Theorem~2.3]{mcdiarmid1998concentration}}]
\label{thm_chernoff}
Let $X_1,\ldots, X_n$ be independent random variables with $0\leq X_i \leq 1$ for each $i$. Let $X=\sum X_i$ and $\overline{X} = \Exp[X]$.
\begin{enumerate}[(1)]
\item For any $\epsilon>0$,
\[
\Prob\left(X \geq (1+\epsilon)\overline{X} \right) \leq e^{-\frac{\epsilon^2\overline{X}}{2(1+\epsilon/3)}}.
\]
\item For any $\epsilon>0$,
\[
\Prob\left(X \leq (1-\epsilon)\overline{X} \right) \leq e^{-\frac{\epsilon^2\overline{X}}{2}}.
\]
\end{enumerate}
\end{theorem}

It will be more convenient to work with random variables for which we only know bounds on the expected value.
For this reason we will use the following corollary of \cref{thm_chernoff}.

\begin{corollary}
\label{cor_chernoff}
Let $X_1,\ldots, X_n$ be independent random variables with $0\leq X_i \leq 1$ for each $i$.
\begin{enumerate}[(1)]
\item For any $\epsilon>0$ and $\overline{X} \geq \Exp[X]$,
\[
\Prob\left(X \geq (1+\epsilon)\overline{X} \right) \leq e^{-\frac{\epsilon^2\overline{X}}{2(1+\epsilon/3)}}.
\]
\item For any $\epsilon>0$ and $\overline{X} \leq \Exp[X]$,
\[
\Prob\left(X \leq (1-\epsilon)\overline{X} \right) \leq e^{-\frac{\epsilon^2\overline{X}}{2}}.
\]
\end{enumerate}
\end{corollary}
\begin{proof}
In order to prove (1) let us consider a random variable $Y=X + Y_1 + Y_2 + \ldots + Y_k$, where $k=\left\lceil \overline{X} - \Exp[X]\right\rceil$ and each $Y_i$ is a constant equal to $\frac{\overline{X} - \Exp[X]}{k}$. Clearly $\Exp[Y]=\overline{X}$ and $Y\geq X$, so the statement follows by \cref{thm_chernoff} (1).

For~(2) it is enough to apply \cref{thm_chernoff}~(2) for the random variable $Y=X\frac{\overline{X}}{\Exp[X]}$.
\end{proof}



We start with a technical lemma that is the crucial ingredient of our algorithm.

\begin{lemma} \label{lemma_technical}
There exists an absolute constant $K$ such that the following is true. Let ${G}$ be a $3$-colorable graph with $n$ vertices such that
\begin{enumerate}[(i)]
\item $\Delta(G)\leq n^{2/3}$,
\item for every $v \in V(G)$, the set $N^{\leq 2}_G(v)$ contains at least $n-\frac{1}{36}n^{2/3}$ vertices,
\item for every two vertices $u,v\in V(G)$ there are at most $n^{2/3}$ vertices $w$ such that $N_G(u) \cap N_G(v) \cap N_G(w) \neq \emptyset$.
\end{enumerate}
Let $\phi$ be a proper $3$-coloring of ${G}$, where $a\in [3]$ is the color that appears most frequently. Define $A := \phi^{-1}(a)$.
Then there exist sets $S \subseteq A$ and $\widetilde{S} \subseteq V(G) \setminus A$, each of size at most $K \cdot n^{1/3} \log n$,
such that $S \cup \widetilde{S} \cup \left( N(S) \cap N(\widetilde{S}) \right)$ dominates at least $\frac{n}{6}$ vertices.
\end{lemma}

Before we prove \cref{lemma_technical}, let us explain its purpose. Suppose that $G$ is a graph with diameter at most $2$ and we are trying to find a $3$-coloring of $G$ under the promise that it exists. We start by assigning to each vertex a list of $3$ possible colors. Note that if we correctly guess a set $S$ of vertices of the most frequent color $a$ and a set $\widetilde{S}$ of vertices together with its coloring using colors $[3]\setminus \lbrace a\rbrace$, then we can deduce the color of each vertex in $N(S) \cap N(\widetilde{S})$.
Hence, our reduction rules will remove at least one color from the list of each vertex dominated by $S \cup \widetilde{S} \cup \left( N(S) \cap N(\widetilde{S}) \right)$. If the sets $S$ and $\widetilde{S}$ are as in the lemma, then we have just removed at least $\frac{n}{6}$ colors from all the lists by guessing the coloring of only $\Oh(n^{1/3}\log n)$ vertices. This is roughly why our algorithm is much faster than an exhaustive search.

The assumptions of the lemma can be read as follows: (i) vertices in $G$ do not have too many neighbors, (ii) $G$ is almost a graph with diameter $2$ and (iii) common neighbors of every two vertices $u$ and $v$ do not dominate too many vertices of the graph. As we will see later, those assumptions arise naturally when trying to solve the problem using simple branching rules -- if any of them is violated, then searching for a $3$-coloring of $G$ becomes easier because of other reasons.

\begin{proof}[Proof of \cref{lemma_technical}]
Note that we can assume that $n \geq n_0$, where $n_0$ is a constant that implicitly follows from the reasoning below.
Indeed, otherwise it is sufficient to set $K := n_0$, $S: = A$, and $\widetilde{S} := V(G) \setminus A$. Thus from now on we assume that $n$ is sufficiently large.

For every two vertices $u,v\in V(G)$ such that $N[u]\cap N[v] \neq \emptyset$, let $x_{uv}$ be a vertex from $N[u]\cap N[v]$.
Fix some vertex $v_a  \in A$ and a function $f:N^{\leq 2}(v_a)\to N(v_a)$ defined such that $f(u)$ is an arbitrarily chosen vertex from $N[u] \cap N(v_a)$.

\begin{figure}
\centering
\includegraphics[scale=1]{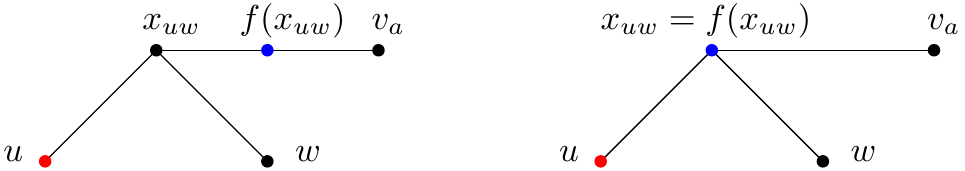}
\caption{The vertex $u$ threatens $w$: if $x_{uw} \in \widetilde{S}$ and $u \in S$, then $w$ has a neighbor with uniquely determined color.}\label{fig:threaten}
\end{figure}

We start by selecting $\widetilde{S}$ as a subset of neighbors of $v_a$. For such a set $\widetilde{S}$ we say that a vertex $u\in A$ \emph{threatens} a vertex $w\in A$ if 
\begin{enumerate}[(1)]
\item $N[u]\cap N[w]\neq \emptyset$, 
\item $x_{uw}\in N^{\leq 2}(v_a)$, and
\item $f(x_{uw})\in \widetilde{S}$.
\end{enumerate}
Intuitively, $u$ threatens $w$ if selecting $u$ to $S$ would undoubtedly cause $w$ to be dominated by $S \cup \widetilde{S} \cup \left(N(S) \cap N(\widetilde{S})\right)$, see~\cref{fig:threaten}.
The following claim gives us a set $\widetilde{S}$ such that each vertex of $A$ is threatened by many vertices.

\begin{claim}
\label{claim_manyThreats}
There exists a set $\widetilde{S}\subseteq N(v_a)$ of order at most $200 n^{1/3} \log n$ such that for at least half of vertices $w\in A$ there are at least $8 n^{2/3} \log n$ vertices from $A$ that threaten $w$. 
\end{claim}
\begin{claimproof}
We select $\widetilde{S}$ randomly in such a way that each neighbor of $v_a$ is included in $\widetilde{S}$ independently with probability $\widetilde{p}=100 n^{-1/3}\log n$. We will show that $\widetilde{S}$ satisfies the desired properties with positive probability.

Note that the size of $\widetilde{S}$ is a sum of $\deg(v_a)$ independent random boolean variables and the expected value of $|\widetilde{S}|$ is $\deg(v_a) \cdot \widetilde{p}$.
Recall that by the assumption (i) we have $\deg(v_a) \leq n^{2/3}$.
Therefore by \cref{cor_chernoff}~(1) applied with $\epsilon=1$ we deduce that
\[
\Prob\left( |\widetilde{S}| > 200 n^{1/3} \log n\right) \leq e^{-37.5 n^{1/3} \log n}.
\]

Let $A^\prime \subseteq A$ be the set of those $v\in A$, for which the set $N\left(N(v)\setminus N^{\leq 2}(v_a)\right)$ contains fewer than half of vertices from $A$.
We will show that $\left|A^\prime\right|\geq \frac{1}{2} \left|A\right|$.
First, let us estimate the number $P$ of ordered pairs of vertices $(u,v)$ such that $u$ and $v$ have a common neighbor outside of $N^{\leq 2}(v_a)$.
By (i) each vertex outside of $N^{\leq 2}(v_a)$ can be a common neighbor for at most $n^{4/3}$ pairs of vertices,
so (ii) implies that $P \leq \frac{1}{36} n^{2}$.
Note that a vertex from $A$ is not contained in $A^\prime$ only if it is in at least $\left|A\right|$ pairs that contribute to $P$.
It follows that $A^\prime$ contains at least $\left|A \right| - \frac{2P}{\left|A\right|}$ vertices. 
Since $a$ is the most frequent color used by the $3$-coloring $\phi$, we have $\left|A\right|\geq \frac{1}{3} n$,
and thus $|A^\prime| \geq \frac{1}{2}\left|A\right|$, as desired.

Fix a vertex $w$ from $A^\prime$. Consider a random variable $X_w$ that counts the number of vertices $u$ from $A$ such that $u$ threatens $w$ and $N(u)\cap N(w)\subseteq N^{\leq 2}(v_a)$. Our plan is to use \cref{cor_chernoff} to show that $X_w$ is at least $8 n^{2/3}\log n$ with high probability. 

We start by estimating the expected value of $X_w$.
Let $U$ be the set of vertices $u$ for $A$ such that $N(u)\cap N(w)\subseteq N^{\leq 2}(v_a)$.
Note that each vertex $u\in U$ contributes $1$ to $X_w$ if and only if $f(x_{uw}) \in \widetilde{S}$, i.e., with probability $\widetilde{p}$.
Since $w\in A^\prime$, the size of $U$ is at least $\frac{1}{2}\left|A\right|$ minus the number of vertices outside of $N^{\leq 2}(w)$, which totals to at least $\frac{n}{6} - \frac{1}{36}n^{2/3}$ by (ii).
Therefore, $\Exp[X_w]\geq 16n^{2/3}\log n$ for large enough $n$.

Now we express $X_w$ as a sum of a number of independent random variables. Fix an ordering $v_1, v_2, \ldots, v_{\deg(v_a)}$ of neighbors of $v_a$ and define $U_i$ as the set of vertices $u$ from $U$ such that $x_{uw}\in N^{\leq 2}(v_a)$
and $f(x_{uw})=v_i$; note that by the definition of $U$, there is a vertex in $N[u]\cap N[w] \cap N^{\leq 2}(v_a)$, so $x_{uw}$ and $f(x_{uw})$ exist for all vertices $u\in U$. For $i=1, 2, \ldots, \deg(v_a)$ let $X_i$ be a random variable that is equal to $\left|U_i\right|$ if $v_i\in \widetilde{S}$ and $0$ otherwise. Clearly $X_w=\sum_i X_i$ and all the variables $X_1,\ldots, X_{\deg(v_a)}$ are independent by the independent selection of $\widetilde{S}$.

By (iii), applied for $w$ and $v_i$, we obtain that $X_i\leq n^{2/3}$ for all $i$. Therefore we may use \cref{cor_chernoff}~(2)
for the sequence of variables $\frac{X_i}{n^{2/3}}$ and $\epsilon=\frac{1}{2}$ to deduce that
\[
\Prob\left(\frac{X_w}{n^{2/3}} \leq 8 \log n \right) \leq e^{-2 \log n},
\]
which gives that
\[
\Prob\left(X_w \leq 8 n^{2/3} \log n \right) \leq n^{-2}.
\]

By the union bound we obtain that the probability that $\widetilde{S}$ has more than $200 n^{1/3}\log n$ vertices or that $X_w<8 n^{2/3} \log n$ for any $w\in A^\prime$ is at most $n \cdot n^{-2} + n^{-37.5 n^{1/3}}$.
Therefore, for large enough $n$ the set $\widetilde{S}$ satisfies the required properties with positive probability, so the proof of the claim is complete.
\end{claimproof}

Having selected $\widetilde{S}$, we proceed to selecting $S$ as a subset of $A$ that guarantees the desired domination property.

\begin{claim}
\label{claim_Domination}
There exists a set $S\subseteq A$ of order at most $2 n^{1/3}$ such that at least half of the vertices $w\in A$ are dominated by $S \cup \widetilde{S} \cup \left( N(S) \cap N(\widetilde{S}) \right)$.
\end{claim}
\begin{claimproof}
We randomly select $S$ so that each vertex from $A$ is in $S$ independently with probability $p=n^{-2/3}$. Note that by \cref{cor_chernoff}~(1) the size of $S$ is at most $2 n^{1/3}$ with probability at least $1 - e^{-\frac{3}{8}n^{2/3}}$.

Let $w$ be a vertex from $A$ that is threatened by at least $8 n^{2/3} \log n$ vertices from $A$. The probability that $w$ is not dominated by $N(S)\cap N(\widetilde{S})$ is at most
\[
\left(1 - p\right)^{8 n^{2/3}\log n} \leq e^{-8p n^{2/3} \log n} \leq e^{-8 \log n} \leq n^{-8}.
\]
By the union bound it follows that with probability at least $1 - n^{-7}$ all vertices threatened by at least $8 n^{2/3} \log n$ vertices from $A$ are dominated by $N(S)\cap N(\widetilde{S})$. \cref{claim_manyThreats} implies that there are at least $\frac{1}{2}\left|A\right|$ such vertices, so the proof is complete.
\end{claimproof}
Setting $K:= \max(n_0,200)$.
Now the statement of the lemma follows from \cref{claim_Domination} by observing that since $A$ is the most frequent color, we have $\frac{1}{2}|A| \geq \frac{1}{6}n$. 
\end{proof}

Now we are ready to prove \cref{thm:main}.

\mainthm*
\begin{proof}
Let $(G,L)$ be an instance of the \lcol{3} problem.
Again, we start by applying reduction rules R1, R2, R3, so we can assume that $(G,L)$ is reduced.
Let $(V_1,V_2,V_3)$ be the layer structure of $G$ and set $\mu:= \left|V_3 \right|$.

We use one of the four branching rules to produce a number of instances of the problem,
each with fewer vertices with lists of size $3$.
Those instances are solved recursively and if a success is reported for at least one of them, then the algorithm terminates and reports a success. The following branching rules are applied in the given order -- it is essential that B4 is executed only if the rules B1, B2 and B3 cannot be applied.
\begin{enumerate}[{B}1]
\item If there exists a vertex $v\in V_2\cup V_3$ such that $v$ has more than $\mu^{2/3}$ neighbors in $V_3$, then for every color $a\in L(v)$ solve an instance obtained by replacing $L(v)$ with $\lbrace a \rbrace$ and exhaustively applying the reduction rules.
\item If there exists a vertex $v \in V_3$ such that for at least $\frac{1}{36} \mu^{2/3}$ vertices $u\in V_3$ a common neighbor of $u$ and $v$ is in $V_2$, then for every color $a\in L(v)$ solve an instance obtained by replacing $L(v)$ with $\{a\}$ and exhaustively applying the reduction rules.
\item If there are two vertices $u,v\in V_3$ such that for at least $\mu^{2/3}$ vertices $w$ from $V_3$ the set $N(u) \cap N(v) \cap N(w)$ is nonempty, then for every two distinct colors $a, b$ construct an instance by setting $L(u):=\lbrace a \rbrace$ and $L(v):=\lbrace b \rbrace$ and one additional instance obtained by replacing vertices $u$ and $v$ with a new vertex $z$ adjacent to $N(u)\cup N(v)$ with $L(z)=[3]$.
Apply the reduction rules to each of those instances and solve them recursively.
\item Let $K$ be the constant from \cref{lemma_technical}. For every tuple $(a,S,\widetilde{S},\phi)$, where
\begin{itemize}
\item $a \in [3]$ is a color,
\item $S \subseteq V_3$ is a set of size at most $K \cdot \mu^{1/3}\log \mu$,
\item $\widetilde{S} \subseteq V_3 \setminus S$ is a set of size at most $K \cdot \mu^{1/3}\log \mu$,
\item $\phi$ is a coloring of $\widetilde{S}$ using colors $[3] \setminus \{a\}$,
\end{itemize}
construct an instance by setting $L(v):=\lbrace a\rbrace$ for each $v\in S$ and $L(v)=\lbrace \phi(v)\rbrace$ for $v\in \widetilde{S}$.
Apply the reduction rules to each of those instances for which $S \cup \widetilde{S} \cup \left( N(S) \cap N(\widetilde{S}) \right)$ dominates at least $\frac{1}{6}\mu$ vertices from $V_3$ and solve them recursively.
\end{enumerate}
Let us show that the above algorithm is correct. 
Branching rules B1 and B2 are clearly correct, because if there is a solution to the given instance of the \lcol{3} problem, then it assigns to $v$ one color from $L(v)$. The rule B3 is correct because if there is a solution to the given instance of the problem, then it either assigns two different colors to $u$ and $v$, or assigns the same color to $u$ and $v$, hence at least one of the constructed instances will admit a solution.  Note that contracting the vertices $u$ and $v$ does not increase the diameter.
Now consider the branching rule B4.
Recall that it is applied only when rules B1, B2 and B3 are inapplicable, so in this case the graph $G[V_3]$ satisfies the assumptions (i)-(iii) of \cref{lemma_technical}.
Therefore if the original instance has a solution, then by \cref{lemma_technical} at least one instance constructed in B4 admits a solution.
On the other hand, each instance is obtained by fixing the colors of vertices in $S \cup \widetilde{S} \subseteq V_3$, so each such a coloring respects lists $L$. Furthermore, if this coloring is improper, then the application of reductions rules R1 and R2 will cause the algorithm to reject the instance. Hence, the branching rule B4 is correct.

Let us denote by $F(x)$ the maximum running time of the algorithm on an instance with at most $x$ vertices with lists of size $3$.
By $\p(n)$ we denote the cost of exhaustively applying the reduction rules to an instance with $n$ vertices; note that $\p(n)$ is polynomial in $n$.

Now we will bound the running time of the algorithm on our instance $(G,L)$ with $\mu$ vertices with lists of size $3$, depending on which branching rule was applied.
\subparagraph*{Case 1: B1 was applied.} Note that this branching produced at most three instances of the problem, each with at most $\mu-\mu^{2/3}$ vertices with lists of size $3$. This is because for every vertex $u\in V_3$ that is a neighbor of $v$ the color $c$ was removed from $L(u)$. Therefore, in this case the running time is at most
\[
3F\left(\mu-\mu^{2/3}\right) + 3\p(n).
\] 
\subparagraph*{Case 2: B2 was applied.} Let $c$ be the color which maximizes the number $N$ of vertices $u\in V_3$ such that the list of a common neighbor of $v$ and $u$ in $V_2$ does not contain $c$; clearly $N\geq \frac{1}{108} \mu^{2/3}$. Let $a$ and $b$ be the two other colors. Note that if a vertex $u$ contributes to $N$, then after the application of reduction rules $b$ (respectively $a$) is removed from $L(u)$ in the instance constructed for the color $a$ (respectively $b$). It follows that the running time of the algorithm in this case is at most
\[
F\left(\mu - 1\right) + 2F\left(\mu-\frac{1}{108}\mu^{2/3}\right) + 3\p(n).
\] 
\subparagraph*{Case 3: B3 was applied.} Let $w$ be a vertex from $V_3$ such that the set $N(u) \cap N(v) \cap N(w)$ is nonempty. Note that if we set $L(u)$ to $\lbrace a \rbrace$ and $L(v)$ to $\lbrace b \rbrace$, for $a \neq b$, then after applying the reduction rules common neighbors of $u$ and $v$ will have lists of size $1$, hence the size of the list of $w$ will be at most $2$. Therefore, in this case the running time is at most
\[
F\left(\mu - 1\right) + 6F\left(\mu-\mu^{2/3}\right) + 7\p(n).
\]
\subparagraph*{Case 4: B4 was applied.} Note that in the constructed instances, after applying the reduction rules, all vertices from $S \cup \widetilde{S} \cup \left( N(S) \cap N(\widetilde{S}) \right)$ have lists of size $1$, so all vertices dominated by $S \cup \widetilde{S} \cup \left( N(S) \cap N(\widetilde{S}) \right)$ have lists of size at most $2$.
Therefore, all instances that are solved recursively have at most $\mu - \frac{1}{6}\mu$ vertices with lists of size $3$. The total number of those instances can be upper bounded by
\[
3 \cdot \mu^{2 K \mu^{1/3}\log \mu}  \cdot 2^{K \mu^{1/3}\log \mu} < 2^{K' \mu^{1/3}\log^2 \mu},
\]
for some constant $K'$.
Therefore the total running time in this case is at most
\[
2^{K'  \mu^{1/3}\log^2 \mu} F\left(\frac{5}{6}\mu\right) + 2^{K' \mu^{1/3}\log^2 \mu}\p(n)
\]

As the considered cases cover all possibilities, we conclude that $F(\mu)$ is bounded by the maximum of the expressions obtained in all four cases. By solving this recurrence we obtain 
\[
F(\mu)\leq \p(n) \cdot 2^{\Oh\left(\mu^{1/3}\log ^2 \mu \right)}=2^{\Oh\left(\mu^{1/3}\log ^2 \mu \right)}.
\]
Since $\mu\leq n$, the proof is complete.
\end{proof}

\section{Possible extensions of our results}\label{sec:ext}
We conclude the paper with discussing possible extensions of our results.

\subsection{Solving \textsc{List $H$-Coloring} in small-diameter graphs}\label{sec:lhom}

For a fixed graph $H$ with possible loops, an instance of \lcol{$H$} is a pair $(G,L)$, where $G$ is a graph and $L : V(G) \to 2^{V(H)}$ is a list function.
We ask whether there exists a \emph{list homomorphism} from $(G,L)$ to $H$, i.e.,  a function $\phi : V(G) \to V(H)$, such that (i) for each $uv \in E(G)$ it holds that $\phi(u)\phi(v) \in E(H)$, and (ii) for each $v \in V(G)$ it holds that $\phi(v) \in L(v)$.
Clearly \lcol{$K_k$} is equivalent to \lcol{$k$}. This is why we refer to the vertices of $H$ as \emph{colors}.

We observe that the algorithm from  \cref{thm:diam3} and \cref{thm:main} can be adapted to \lcol{$H$} if the graph $H$ satisfies certain conditions.
First, the algorithm from \cref{thm:diam3} can be adapted to solve the \lcol{$H$} problem if
\begin{enumerate}[({P}1)]
\item every vertex of $H$ has at most two neighbors (possibly including itself, if it is a vertex with a loop). \label{prop:p1}
\end{enumerate}
For such graphs $H$, once we fix a color of some $v \in V(G)$, all its neighbors have lists of size at most 2.

To adapt the algorithm from \cref{thm:main}, in addition to property (P\ref{prop:p1}), we need two more:
\begin{enumerate}[({P}1)]
\setcounter{enumi}{1}
\item any two distinct vertices of $H$ must have at most one common neighbor, \label{prop:p2}
\item $H$ has no loops. \label{prop:p3}
\end{enumerate}
Property (P\ref{prop:p2}) is needed to ensure that as soon as we fix the coloring of the sets $S$ and $\widetilde{S}$ selected in \cref{lemma_technical}, then the color of every vertex in $N(S) \cap N(\widetilde{S})$ is uniquely determined.
Property (P\ref{prop:p3}) is needed for our selection of the set $\widetilde{S}$: recall that all these vertices are in the neighborhood of some vertex $v_a$ colored $a$, which is sufficient to ensure that no vertex of $\widetilde{S}$ gets the color $a$.

Let $\mathcal{H}$ be the family of connected graphs that satisfy property (P\ref{prop:p1}).
From the complexity dichotomy for \lcol{$H$} by Feder, Hell, and Huang~\cite{DBLP:journals/combinatorica/FederHH99,DBLP:journals/jgt/FederHH03} it follows that if $H \in \mathcal{H}$, then \lcol{$H$} is polynomial-time solvable if:
\begin{itemize}
\item $H$ has at most two vertices,
\item $H = C_4$,
\item $H$ is a path,
\item $H$ is a path with a loop on one endvertex,
\end{itemize}
and otherwise the problem is \NP-complete and does not admit a subexponential-time algorithm under the ETH.
So, in other words, there are two families of graphs $H \in \mathcal{H}$ for which the problem is \NP-complete (in general graphs):
\begin{itemize}
\item all cycles $C_k$ for $k =3$ or $k \geq 5$, and
\item all graphs obtained from a path with $k \geq 3$ vertices by adding loops on both endvertices; let us call such a graph $P^*_k$.
\end{itemize}
Let us present one more simple observation about solving \lcol{$H$} in graphs with small diameter.
Consider an instance $(G,L)$ of \lcol{$H$} and suppose that $H$ contains two vertices $x,y$ at distance greater than $\diam(G)$.
(Here, with a little abuse of notation, we use the convention that if $x$ and $y$ are in different connected components of $H$, then their distance is infinite.)
We note that there is no (list) homomorphism from $G$ to $H$ that uses both $x$ and $y$.
Thus we can reduce the problem to solving an instance $(G,L_x)$ of \lcol{$(H - x)$} and an instance $(G,L_y)$ of \lcol{$(H-y)$}, where lists $L_x$ (resp. $L_y$) are obtained from $L$ by removing the vertex $x$ (resp., $y$) from each set.

Combining all observations above, we obtain the following results. We skip the formal proofs, as they are essentially the same as the ones of \cref{thm:diam3} and \cref{thm:main} and bring no new insight.
\begin{theorem}
Let $H \in \mathcal{H}$. Consider an instance $(G,L)$ of \lcol{$H$}, where $G$ is of diameter $2$.
Then $(G,L)$ can be solved 
\begin{enumerate}
\item in polynomial time if $H \notin \{C_3,C_5,P^*_3\}$,
\item in time $2^{\Oh(n^{1/3} \log^{2} n)}$ if $H \in \{C_3,C_5\}$,
\item in time $2^{\Oh(n^{1/2} \log^{1/2} n)}$ if $H = P^*_3$.
\end{enumerate}
\end{theorem}
\begin{theorem}
Let $H \in \mathcal{H}$. Consider an instance $(G,L)$ of \lcol{$H$}, where $G$ is of diameter $3$.
Then $(G,L)$ can be solved 
\begin{enumerate}
\item in polynomial time if $H \notin \{C_3,C_5,C_6,C_7,P^*_3,P^*_4\}$,
\item in time $2^{\Oh(n^{2/3} \log^{2/3} n)}$ if $H \in \{C_3,C_5,C_6,C_7,P^*_3,P^*_4\}$.
\end{enumerate}
\end{theorem}

\subsection{Weighted coloring problems}

Another possible generalization of \lcol{3} would be to introduce weights: for each pair $(v,c)$, where $v \in V(G)$ and $c \in \{1,2,3\}$, we are given a cost $\wei(v,c)$ of coloring $v$ with $c$, and we ask for a proper coloring minimizing the total cost. A natural special case of this problem is \textsc{Independent Odd Cycle Transversal}, where we ask for a minimum-sized independent set which intersects all odd cycles.

Let us point out that the branching phases in our algorithms from \cref{thm:diam3} and \cref{thm:main} can handle this type of modification.
However, this is no longer the case for the last phase, when the problem of coloring a graph with all lists of size at most two is reduced to 2-\textsc{Sat} using \cref{thm:edwards}. It is known that a weighted variant of 2-\textsc{Sat} is \NP-complete and admits no subexponential-time algorithm, unless the ETH fails~\cite{DBLP:journals/amai/Porschen07}.
Thus, in order to extend our algorithmic results to weighted setting, we need to find a way to replace using \cref{thm:edwards} with some other strategy of dealing with lists of size 2.

\bibliography{main}
\end{document}